\documentclass[11pt]{article}
\usepackage{color}
\usepackage[colorlinks=true, allcolors=blue,backref=page]{hyperref}
\usepackage{amsmath, amssymb, amsthm}
\usepackage{mathrsfs}
\usepackage{mathtools}
\usepackage[noabbrev,capitalize,nameinlink]{cleveref}
\usepackage{fullpage}
\usepackage[noadjust]{cite}
\usepackage{graphics}
\usepackage{pifont}
\usepackage{tikz}
\usepackage{bbm}
\usepackage{float} 
\usepackage[T1]{fontenc}

\usetikzlibrary{arrows.meta}

\usepackage{environ}
\usepackage{framed}
\usepackage{url}
\usepackage[linesnumbered,ruled,vlined]{algorithm2e}
\usepackage[noend]{algpseudocode}
\usepackage[labelfont=bf]{caption}
\usepackage{cite}
\usepackage{framed}
\usepackage[framemethod=tikz]{mdframed}
\usepackage{appendix}
\usepackage{graphicx}
\usepackage[textsize=tiny]{todonotes}
\usepackage{tcolorbox}
\usepackage{enumerate}
\allowdisplaybreaks[1]
\usepackage{enumerate}
\usepackage{stmaryrd}
\usepackage[margin=1in]{geometry}

\usepackage[shortlabels]{enumitem}
\crefformat{enumi}{#2#1#3}
\crefrangeformat{enumi}{#3#1#4 to~#5#2#6}
\crefmultiformat{enumi}{#2#1#3}
{ and~#2#1#3}{, #2#1#3}{ and~#2#1#3}

\DeclareSymbolFont{symbolsC}{U}{pxsyc}{m}{n}
\SetSymbolFont{symbolsC}{bold}{U}{pxsyc}{bx}{n}
\DeclareFontSubstitution{U}{pxsyc}{m}{n}
\DeclareMathSymbol{\medcircle}{\mathbin}{symbolsC}{7}

\crefname{equation}{}{} 
\AtBeginEnvironment{appendices}{\crefalias{section}{appendix}} 

\usepackage[color,final]{showkeys} 

\colorlet{refkey}{orange!20}
\colorlet{labelkey}{blue!30}

\numberwithin{equation}{section}
\newtheorem{theorem}{Theorem}[section]

\newtheorem{lemma}[theorem]{Lemma}
\newtheorem{claim}[theorem]{Claim}

\newtheorem{corollary}[theorem]{Corollary}

\newtheorem*{question*}{Question}

\theoremstyle{definition}
\newtheorem{definition}[theorem]{Definition}

\newtheorem*{definition*}{Definition}
\newtheorem{example}[theorem]{Example}

\theoremstyle{remark}

\newcommand{\mb}{\mathbb}
\newcommand{\mbf}{\mathbf}

\newcommand{\mc}{\mathcal}

\newcommand{\on}{\operatorname}

\newcommand{\bx}{\mathbf{x}}
\newcommand{\by}{\mathbf{y}}
\newcommand{\bp}{\mathbf{p}}

\let\originalleft\left
\let\originalright\right
\renewcommand{\left}{\mathopen{}\mathclose\bgroup\originalleft}
\renewcommand{\right}{\aftergroup\egroup\originalright}

\allowdisplaybreaks

\newif\ifpublic
\publictrue

\ifpublic

\newcommand{\ignore}[1]{}

\else

\fi

\newcommand{\cV}{{\mathcal V}}
\newcommand{\cC}{{\mathcal C}}
\newcommand{\cD}{{\mathcal D}}

\title{Computation of Approximately Stable Committees in Approval-based Elections}
\author{Drew Gao \\ Stanford University \and Yihang Sun \\ Stanford University \and Jan Vondr\'{a}k \\ Stanford University}
\begin{document}

\maketitle
\begin{abstract}
    Approval-based committee selection is a model of significant interest in social choice theory. In this model, we have a set of voters $\cV$, a set of candidates $\cC$, and each voter has a set $A_v \subset \cC$ of approved candidates. For any committee size $K$, the goal is to choose $K$ candidates to represent the voters' preferences. We study a criterion known as \emph{approximate stability}, where a committee is $\lambda$-approximately-stable if there is no other committee $T$ preferred by at least $\frac{\lambda|T|}{k} |\cV| $ voters. We prove that a $3.65$-approximately stable committee always exists and can be computed algorithmically in this setting.    
    Our approach is based on finding a Lindahl equilibrium and sampling from a strongly Rayleigh distribution associated with it.
    \end{abstract}

\section{Introduction} 
The field of social choice theory is concerned with how individual preferences can be aggregated into a collective decision via a \emph{social choice function} (or \emph{voting rule}). Social choice theory seeks to formalize notions such as fairness and proportional representation, and to answer questions such as: How can we formally define a fair or optimum outcome? How can we ensure that no subset of voters feels that they have not been represented in the outcome? What assumptions need to be made on the structure of preferences in order for a desirable outcome to exist or be efficiently computable? 

In the classical single-winner setting, this is modeled as a collection of \emph{voters} $\cV$ who must choose one winner from a set of \emph{candidates} $\cC$ (also referred to as \emph{alternatives}), and each voter reports a ranked preference order over the candidates. These preference orders are then used as input to some \emph{voting rule}, a function which outputs a winning candidate in $\cC$. Unfortunately, classical impossibility results (e.g. Arrow's Impossibility Theorem \cite{arrow1950difficulty}) show the non-existence of social choice functions satisfying a collection of seemingly natural fairness criteria. These negative results for single-winner elections have motivated models where approximate guarantees can be proved in a certain sense. While we do not attempt to survey this substantial body of literature here, let us mention two notable recent papers for single-winner elections: choice of a single winner in metric spaces with only ordinal preferences available to the algorithm \cite{CRWW24}, and a choice of a list of $6$ candidates that in a certain sense beats any single candidate \cite{CLRV24}.


The single-winner model generalizes to scenarios where voters choose a subset of candidates as a committee rather than a single winner. This is often referred to as the \emph{committee selection problem}. A related model called \emph{participatory budgeting} is where alternatives have variable weights (costs of ``projects'') and voters want to choose a collection of projects within a certain budget. Participatory budgeting problems can be considered with divisible or indivisible items or projects. The committee selection problem can be viewed as a participatory budgeting problem with indivisible items, under a cardinality constraint.

\subsection{ABC elections.}
We consider committee selection problems where we choose a committee $S$ from a pool of candidates $\cC$ to represent a population of voters $\cV$. In this paper, we focus on the model of {\em Approval-Based Committee} elections, or ABC-elections. 

\begin{definition}[ABC Elections]
\label{def:abc}
$(\cV, \cC, \{A_v\}_{v\in \cV})$ is an instance of an \emph{approval-based committee (ABC) election}, if
$\cV$ is a finite set of \emph{voters},  $\cC$ is a finite set of \emph{candidates}, and $A_v \subset \cC$ is an \emph{approval set} for every $v\in \cV$. This induces a preorder $\succ_v$ (with possible ties) on $2^\cC$ for every voter $v$: we say $v$ \emph{strictly prefers $T$ to $S$}, or $T\succ_v S$, if
\[ \left|A_v\cap T\right|>\left|A_v\cap S\right|\]
An outcome of the election is a \emph{committee} $S\subset \cC$.
\end{definition}

More generally, $\succ_v$ could be any preorder on $2^{\cC}$ (usually assumed to be monotonic with respect to $\subset$). We do not consider the problem in this generality here.

Given a preference profile $\{A_v\}_{v \in \cV}$, we would like to select a committee of cardinality $K$ which is reflective of all voters preferences.  One of the oldest and most studied notions is that of \emph{core-stability}, which captures the idea that any coalition containing a certain proportion of voters, should be represented by the same proportion of candidates within the chosen committee. 

\begin{definition}[Core Stability]
In an ABC election $(\cV, \cC, \{A_v\})$ with $n=|\cV|$ voters, we say a committee $S\subset \cC$ is \emph{blocked} by a \emph{deviating committee} $T\subset \cC$ if
\begin{equation}
\label{eq:def-core-ord}
\left|\{v\in\cV:T\succ_v S\}\right| \geq \frac{|T|n}{|S|}
\end{equation}
and we call $\{v\in\cV:T\succ_v S\}$ the \emph{deviating coalition}. We say $S$ is \emph{core-stable} (or \emph{stable}) if it is not blocked by any other committee, i.e. for every $T\subset \cC$ we have
\begin{equation}
\label{eq:def-core}
\left|\{v\in\cV: T \succ_v S \}\right|
< \frac{|T|n}{|S|}
\end{equation}
The \emph{core} is the set of all core-stable committees.
\end{definition}

Let us illustrate the definition with a simple example.

\begin{example}
Let $\cC=\{a, b, c, d, e\}$, $\cV=[5]$, $k=3$ and $A_1=A_2=A_3=\{a, b, c\}, A_4=A_5=\{d, e\}$. Candidates $a,b,c$ are in some sense the most popular ones, but $S = \{a,b,c\}$ is not a stable committee, since $2$ voters out of $5$ prefer $T = \{ d \}$, and $\frac{|T|}{|S|} = \frac{1}{3}$. This illustrates the fact that a core-stable committee should represent every subset of voters, and in this case voters $\{4,5\}$ are not represented. There are multiple stable committees here, for example $S' = \{a,b,d\}$. 
\end{example}

The most significant open question in the study of ABC elections is whether the core is always non-empty. This is still open as of writing of this paper.

Given this difficulty, several approximate versions of core-stability have been proposed. We state the notion we study in this paper, and consider some related notions in \cref{sec:related-work}.

\begin{definition}[$\lambda$-approximate stability]
\label{def:approx-stable}
In an ABC election $(\cV, \cC, \{A_v\})$ with $n=|\cV|$ voters, we say a committee $S\subset \cC$ is \emph{$\lambda$-approximately stable} for $\lambda \ge 1$ if for every $T\subset \cC$
\begin{equation}
\label{eq:def-approx-core}
\left|\{v\in\cV:T \succ_v S \}\right| <
 \frac{\lambda |T|n}{|S|}.
\end{equation} 
\end{definition}

Sometimes this is also referred to as ``group-size approximate stability''. 
When $\lambda = 1$, this exactly recovers core-stability. 
In a major result, \cite{jiang2020approximately} proved that a $32$-approximately stable committee of given size exists for any election instance with monotone preferences. For approval-based (ABC) elections, their approach in fact implies a factor of $16$. However, their result was only existential, even for approval-based elections.

\subsection{Our Contributions and Techniques}
\label{sec:pf-sketch}
Our main result is the following improvement for approval-based elections. Importantly, the result is algorithmic, meaning we can compute the respective committee in polynomial time.

\begin{theorem}\label{thm:main}
For any ABC election and a given committee size $K$, there exists a $3.651$-approximately stable committee. Moreover, it can be computed in randomized polynomial time.
\end{theorem}

In parallel work, \cite{nguyen-song} improved the approximate stability factor to $6.24$ for monotone preferences and $4.97$ for ABC elections. Moreover, \cite{nguyen-song} shows how to compute $9.94$-approximately stable committees for ABC elections in randomized polynomial time. We discuss the comparison between our works further below.

At a high level, we follow the paradigm of \cite{jiang2020approximately} where a fractional committee is constructed first, and then converted into a feasible committee via randomized rounding. Where we deviate from \cite{jiang2020approximately} is that instead of a lottery over committees satisfying guarantees in expectation, we consider a fractional solution determined by a {\em Lindahl equilibrium}. 
The Lindahl equilibrium has appeared in related work such as \cite{munagala2021approximate}, although not for the type of approximate stability considered here. 
The Lindahl equilibrium has the appealing property that it satisfies core stability in a fractional sense.
Moreover, a Lindahl equilibrium can be computed algorithmically thanks to a recent work of \cite{compute-lindahl}. 

A natural approach to converting the Lindahl equilibrium into a discrete solution is randomized rounding. However, an issue in constructing an approximately stable committee according to \cref{def:approx-stable} is that the guarantee for each voter is very brittle; a voter is dissatisfied with our committee even if a competing committee has one extra approved candidate. A different, utility-based type of approximation was studied by \cite{munagala2021approximate}; their results are orthogonal to ours, as we discuss further below.

We achieve the stability guarantee by comparing the quality of the rounded solution to a Lindahl equilibrium, and making sure that every voter is eventually at least as satisfied as they are in some Lindahl equilibrium. For this purpose, it is very convenient here to
consider {\em strongly Rayleigh distributions} \cite{borcea2008negativedependencegeometrypolynomials} arising from the Lindahl equilibrium (or more precisely, from the Lindahl equilibrium scaled by a factor $\alpha > 1$). These distributions satisfy strong negative correlation and concentration properties. More precisely, the utility function $u_v(S) = |S \cap A_v|$ for the rounded solution behaves like a generalized binomial distribution (summation of independent Bernoulli variables). This allows us to prove very tight tail bounds on the satisfaction of each voter, and compare the quality of the rounded solution to the Lindahl equilibrium. We remove the voters who are satisfied with our rounded solution at least as they are with the Lindahl equilibrium, since their contribution to a deviating coalition can be bounded appropriately. 

Furthermore, we include a greedy phase to improve the satisfication of some voters whose utility $u_v(S) = |S \cap A_v|$ is just below the utility they derive from the Lindahl equilibrium, $u_v(\bx) = \sum_{i \in A_v} x_i$. As long as many such voters approve the same candidate, we include such a candidate in our solution. This simple modification improves the performance of our algorithm significantly.
For voters, who are still not satisfied with the committee that we constructed, the solution is completed recursively. This is analogous to the high-level approach of \cite{jiang2020approximately}.
We provide a high-level description of our algorithm below to find an approximately stable committee of size $K$ in an ABC instance.

\begin{algorithm}
\renewcommand{\thealgocf}{}
\DontPrintSemicolon
\caption*{High Level Sketch of the Main Algorithm}
\label{alg:alg-sketch}

Choose $k$ depending on $K$, and compute a Lindahl equilibrium $(\mbf{x}, \mbf{p})$ such that $\Vert \mbf{x}\Vert_1=k$\;
Choose $\alpha\ge 2$ and construct a strongly Rayleigh distribution $\cD$ on $2^{\cC}$ with marginal probabilities $x'_i = \min \{ \alpha x_i, 1\}$ \;
Sample $R \sim \cD$ until we find $R$ that dominates $\mbf{x}$ approximately in some technical sense \;
$\cV_1 \gets$ voters who prefer $\mbf{x}$ to $R$ by at least one candidate \;
$\cV_2 \gets$ voters who prefer $\mbf{x}$ to $R$ by at least two candidates \;
For suitable parameters $\beta, \gamma\ge 0$ (possibly depending on $R$),
choose additional candidates greedily: as long as there is a candidate approved by at least $\beta k/n$ additional voters in $\cV_1 \setminus \cV_2$,  remove such voters from $\cV_1$ and include the candidate in $R$ \;
If we find $\gamma k$ such candidates, recurse on $\cV_1$ to complete the committee to size $K$; else, recurse on $\cV_2$ to complete the committee to size $K$.
\end{algorithm}


\paragraph{A note on the chronology of recent work.}
An earlier (unpublished) version of this paper, gave a weaker approximation factor of $11.66$ for ABC elections. The approach of \cite{nguyen-song}, providing a factor of $4.97$ for ABC elections and $6.24$ for general monotonic preferences was developed independently. The main difference between our approach and that of \cite{nguyen-song} is a different notion of a Lindahl equilibrium. Our approach uses the more standard definition of a Lindahl equilibrium. Due to a very recent work of Kroer and Peters \cite{compute-lindahl}, this equilibrium can be computed in polynomial time, which was not known at the time of our previous version. 

The equilibrium concept proposed in \cite{nguyen-song}, dubbed $\alpha$-LEO, is more involved but can be extended to more general models. To our knowledge, it is not known how to compute $\alpha$-LEO at the moment. In the case of ABC elections, the authors of \cite{nguyen-song} replace it by another relaxation which can be solved algorithmically, but incurs an additional factor of $2$, leading to a $9.94$-approximately stable committee. 
Recent improvements in our approach (discussed above) decreased the factor for ABC elections elections to $3.651$. Another technicality is that prior work often ignores rounding issues and thus the results are only valid asymptotically as $k \to \infty$. In this work, we handle rounding issues carefully and our results are valid for all $k$. The factor can be improved to $3.61$ if we only care about asymptotic performance.

\subsection{Other Related Work}
\label{sec:related-work}

It has been verified that the core is nonempty for all ABC elections with committee size $K \leq 8$, or number of candidates $m \leq 15$  \cite{peters2025core}. For $K \leq 7$, a stable committee can be found using the classical PAV rule \cite{peters2025core}. However, it is known that this  rule does not work in general, and in fact there is no rule using only the welfare vector ${\mathbf w}=(|S \cap A_v|)_{v\in \cV}$ guaranteed to find a core-stable commmittee \cite{peters2022proportionalitylimitswelfarism}. Additionally, it is known that for other classes of preferences such as ranking-based, the core can be empty \cite{Cheng_2019,jiang2020approximately}. 

ABC elections is a special case of more general models. One can define a utility function for each voter, in the case of ABC elections $u_v(S) = |S \cap A_v|$, and then define a preference order on $2^\mc{C}$ via $T \succ_v S$ if and only if $u_v(T) >u_v(S)$. Other popular choices of utility classes are: (i) additive, i.e. $u_v(S) = \sum_{i \in S} w_{vi}$ for $w_{vi} \geq 0$; (ii) ranking, usually defined via an ordering of candidates or equivalently $u_v(S) = \max_{i \in S} w_{vi}$; (iii) submodular, i.e. $u_v(S) + u_v(T) \geq u_v(S \cup T) + u_v(S \cap T)$, generalizing both of the previous two settings; and (iv) arbitrary monotonic functions $u_v(S)$.

As mentioned, a $32$-approximately-stable committee was proved to exist for any monotone preferences \cite{jiang2020approximately}. As an intermediate concept, it uses the notion of an approximately stable lottery.
 
\begin{definition}[$\lambda$-Approximately Stable Lottery]
In an ABC election with $n$ voters and target committee size $K$, a distribution $\Delta$ over committees is an \emph{$\lambda$-approximately stable lottery} if for any committee $T$ we have that $$\mb{E}_{S\sim\Delta}[|\{ v \in \cV: T \succ_v S\}|]< \frac{\lambda|T|n}{K}.$$ 
A \emph{stable lottery} is one satisfying this condition with $\lambda=1$.
\end{definition}

A stable lottery exists for ABC elections \cite{Cheng_2019}; however, this result is not algorithmic. More generally, \cite{jiang2020approximately} proved that a $2$-approximately stable lottery exists for every instance with monotonic preference orders, and this is the foundation of their result that a $32$-approximately stable committee always exists in this setting (again, this is only an existential result). We remark that by combining the stable lottery result for ABC election with the rounding technique of \cite{jiang2020approximately}, one can obtain a 16-approximately stable committee for ABC elections, although this is not explicitly stated in  \cite{jiang2020approximately}. The question of finding approximately stable committees for ABC elections algorithmically was left an open question in \cite{jiang2020approximately}.

\subsubsection*{Utility-based approximations.}
Let us mention briefly another approximation concept which has been studied in the context of utility-based elections. (We do not pursue this direction in this paper.)

\begin{definition}($\alpha$-satisfaction-approximate stability)
A committee $S \subseteq \cC$ is $\alpha$-satisfaction-approximately stable if there is no $T\subset \cC$ such that $u_i(T) > \alpha u_i(S)$ for a $\frac{|T|}{|S|}$-fraction of voters. A committee $S \subseteq \cC$ is $\alpha$-satisfaction-approximately stable with additament, if there is no $T\subset \cC$ such that $u_i(T) > \alpha u_i(S \cup \{q\})$ for every $q\in \cC$ for a $\frac{|T|}{|S|}$-fraction of voters.
\end{definition}

Note that here, the approximation factor applies to the utility function rather than the size of the deviating coalition. For this notion of approximation, it is known that a $2$-satisfaction-approximately stable committee exists for ABC elections \cite{peters2022proportionalitylimitswelfarism}. This result is based on a discrete local search using the ``PAV rule'' (dating back to 1895!) and is algorithmically efficient. More generally, for the case of submodular utilities, \cite{munagala2021approximate} proves that there is a $67.37$-satisfaction-approximately stable committee, and it can be computed in polynomial time. 
The same work also improve this factor to $9.27$ for additive utility functions.

Other criteria have been defined to formalize the idea of proportional representation committee elections. A committee $S$ satisfies {\em justified representation} if for any group of voters of size at least $\frac{n}{k}$ who all share at least $1$ approved candidate, each voter in this group has at least $1$ approved candidate in $S$. Justified representation, as well as the related but stronger notions of extended justified representation and proportional justified representation, are satisfied by any committee in the core. However, unlike the core, committees satisfying proportional/justified representation are known to exist in all ABC elections.

\section{Preliminaries}
\subsection{Lindahl Equilibrium}

An important concept related to stable committees is that of a {\em Lindahl equilibrium}. This can be viewed as a fractional solution $(x_i: i \in \cC)$ with prices $(p_{vi}: v \in \cV, i \in \cC)$ that certify the stability of $\mbf{x}$ in a fractional sense. This concept is defined in settings with utility functions $u_v:\mb{R}^\cC \to \mb{R}$ for each voter $v \in \cV$; in the ABC setting, this utility can be naturally defined as $u_v({\mathbf x}) = \sum_{j \in A_v} x_j$.

\begin{definition}[Lindahl equilibrium] 
Given an instance of utility-based committee election with target committee size $k$ and utility functions $u_i$, a Lindahl equilibrium is a pair $({\mathbf x}, {\mathbf p})$ where ${\mathbf x} \in  [0,1]^{\cC}$ and ${\mathbf p} \in  \mb{R}_+^{\cV \times \cC}$ such that $\| \mathbf{x} \|_1 = k$, $\mathbf{p}_v \cdot \mathbf{y} \leq \frac{k}{n}$ for each $v \in \cV$, and
	\begin{itemize}
		\item For every voter $v \in \cV$, their utility $u_v(\mathbf{y})$ is maximized subject to $\mathbf{y} \in [0,1]^\cC$ and $\mathbf{p}_v \cdot\mathbf{y} \leq \frac{k}{n}$ when $\mathbf{y}=\mathbf{x}$.
		\item The global profit, $\sum_{v\in{\cV}}\mathbf{p}_v\cdot {\mathbf z} - \|\mathbf{z}\|_1$ is maximized  subject to $\mathbf{z}\geq 0$ when $\mathbf{z}=\mathbf{x}.$
	\end{itemize}
\end{definition}

Additional properties that a Lindahl equilibrium is known to satisfy are the following \cite{munagala2021approximate}.

\begin{lemma}
\label{lem:Lindahl}
A Lindahl equilibrium $({\mathbf x}, {\mathbf p})$ satisfies the following properties:
\begin{itemize}
\item For each voter $v \in \cV$, $\mathbf{p}_v \cdot \mathbf{x} = \frac{k}{n}$.
\item For every $i \in \cC$, $\sum_{v \in \cV} p_{vi} \leq 1$, with equality for every $i \in \cC$ such that $x_i > 0$. 
\end{itemize}
\end{lemma}

It is a classical result that a Lindahl equilibrium always exists for stricly monotone and concave utilities \cite{Foley1970LindahlCore}. This result was notably not constructive, due to an application of a fixed point theorem. However, very recently this result was made algorithmic in the setting that we care about, for linear utilities in the ``capped setting'' as $x_i \in [0,1]$.

\begin{theorem}[\cite{compute-lindahl}]
\label{thm:lindahl-compute}
A Lindahl equilibrium always exists and can be found in polynomial time for committee elections with linear utilities $u_v(\bx) = \sum_{i \in \cC} w_{vi} x_i$, $w_{vi} \geq 0$, and funding caps $x_i \in [0,1]$.
\end{theorem}

We refer the reader to Theorem 7 in \cite{compute-lindahl}; we note that we do not need the zero-respecting property here. A key property of the Lindahl equilibrium is that it is core-stable in a fractional sense. We include the proof here since it is short and instructive.

\begin{theorem}[\cite{munagala2021approximate}]
\label{thm:lindahl-stable}
Every Lindahl equilibrium $(\bx, \mathbf{p})$ is fractionally core-stable, meaning that for any $\by \in [0,1]^{\cC}$, the fraction of voters for whom $u_v(\by) > u_v(\bx)$ is less than $\| \by \|_1 / \| \bx\|_1$.
\end{theorem}

\begin{proof}
We can assume that $\|\by\|_1 < \|\bx\|_1 = k$, otherwise the statement is trivial. Let $\cV_{pref}$ be the subset of voters who prefer $\by$ to $\bx$, i.e. $u_v(\by) > u_v(\bx)$. By the voter-optimality criterion, we must have $\bp_v \cdot \by > \bp_v \cdot \bx$ for $v \in \cV_{pref}$. Then, using the fact that $\sum_{v \in \cV} p_{vi} \leq 1$, we obtain
$$ \| \by \|_1 \geq \sum_{v \in \cV} \bp_v \cdot \by
> \sum_{v \in \cV_{pref}} \bp_v \cdot \bx = |\cV_{pref}| \frac{k}{n}.$$
Hence $\frac{1}{n} |\cV_{pref}| < \frac{1}{k} \| \by \|_1 = \frac{\|\by\|_1}{\|\bx\|_1}$.
\end{proof}

\subsection{Strongly Rayleigh distributions under a cardinality constraint}

In this section, we review some tools from the theory of real stable polynomials, implying the following key result that allows us to round (scaled) Lindahl equilibrium solutions to integer solutions. The results in this section are well known to experts, but we could not find them stated explicitly in the form that we need, hence we summarize them for completeness.

\begin{theorem}
\label{thm:sr-us}
For every $\bx \in [0,1]^\cC$ with integer $\kappa = \|\bx\|_1$, there is a (strongly Rayleigh) distribution $\cD$ supported on subsets of $\cC$ of cardinality $\kappa$ such that for $X = \mbf{1}_R, R \sim \cD$, we have $\mb{E} X_i = x_i$, and for any fixed $F\subset \cC$, $\sum_{i\in F} X_i$ has a generalized binomial distribution (equivalent to a sum of independent Bernoulli random variables).
\end{theorem}

Note that the variables $X_i$ are not independent; hence it is quite surprising that $\sum_{i \in F} X_i$ has a generalized binomial distribution.
We obtain this distribution by considering the {\em maximum entropy} distribution with the required marginal probabilities $x_i$. It turns out that this distribution has the {\em strongly Rayleigh property}, related to stability of multivariate polynomials.

\begin{definition}(Generating Polynomial)
Let $\mu:2^{[m]} \to [0,1]$ be a probability distribution over subsets of $[m]$. The \emph{generating polynomial} of $\mu$ which we denote $g_{\mu}$ is given by $$g_{\mu}=\sum_{S\in supp(\mu)}\mu(S)z^S=\mu(S) \prod_{j \in S} z_j.$$
\end{definition}

\begin{definition}[Real Stable Polynomial]
    A polynomial $p(z_1,\ldots,z_m)$ is \emph{real stable} if it has real coefficients and $p(z_1,z_2,\ldots,z_m)\neq 0$ for every $(z_1,\ldots,z_m)$ such that $Im(z_j)>0$ for all $j$.  
\end{definition}

\begin{definition}[Strongly Rayleigh distribution]
    A probability distribution $\mu$ is \emph{strongly Rayleigh} if its generating polynomial is real stable. 
\end{definition}

As a motivating example of a strongly Rayleigh distribution, we consider the uniform distribution over subsets of $k$ out of $m$ elements, which we denote $\Delta_{k,m}$. 
The following result is well known \cite{Gharan_MaxEntropy}.

\begin{lemma}
\label{lem:elem-stable}
    The elementary symmetric polynomial, defined as $$e_k(z_1,\ldots,z_m)=\sum_{S \subset [m]: |S|=k}\prod_{j \in S}^k z_j$$ is real stable.
\end{lemma}

Since the generating polynomial of $\Delta_{k,m}$ is given by $e_k(z_1,\ldots,z_m)$ with some constant normalizing factor, we have that $\Delta_{k,m}$ is a strongly Rayleigh distribution.
More generally, we have the following.

\begin{lemma}[Maximum-entropy distribution]
  Given $\bx \in (0,1)^m$ such that $\|\bx\|_1 = \kappa$, there are parameters $\lambda_1,\ldots,\lambda_m>0$ such that the distribution defined by
  $$ \mu(S) = \frac{\prod_{i \in S} \lambda_i}{\sum_{T:|T|=\kappa} \prod_{j \in T} \lambda_j} $$
  for all $|S| = \kappa$ satisfies $\sum_{S: |S|=\kappa, j\in S} \mu(S) = x_i$, and the distribution is strongly Rayleigh.
\end{lemma}

\begin{proof}[Proof Sketch.]
The parameters $\lambda_j>0$ can be found by maximizing the entropy function 
$$ H(\lambda_1,\ldots,\lambda_m) =  \sum_{i=1}^{m} x_i \log  \lambda_i - \log \sum_{|S|=\kappa} \prod_{i \in S} \lambda_i $$
over $\lambda_1,\ldots,\lambda_m>0$; the reader can verify that the KKT conditions imply that the marginal probabilities are $x_i$ as desired. 

The strong Rayleigh property follows from the fact that (up to scaling) the generating polynomial of this distribution is a variant of the elementary symmetric polynomial, $$p(z_1,\ldots,z_m)=\sum_{S \subset [m], |S|=\kappa} \prod_{j \in S} \lambda_j z_j,$$ which is real stable by \cref{lem:elem-stable} (and the scaling of $z_j$ by $\lambda_j>0$ does not affect stability).
\end{proof}

We note that although the above lemma assumes $\bx \in (0,1)^m$, variables of integer values can be handled easily since these elements appear deterministically (never or always). 

The final point is that projecting such a distribution onto one dimension results in a generalized binomial distribution.
This follows from the closure properties of stable polynomials \cite{borcea2008negativedependencegeometrypolynomials}.

\begin{lemma}\label{closeure-of-SR}
If $p(z_1,\ldots,z_m)$ is a real stable polynomial and $F \subset [m]$, then $q(z) = p(z_1,\ldots,z_m)$ where we substitute $z_j = z$ for $j \in F$ and $z_j = 0$ for $j \notin F$ is a univariate real-rooted polynomial.
\end{lemma}

\begin{lemma}
A univariate distribution whose generating polynomial is real-rooted with nonnegative coefficients is a generalized binomial distribution.
\end{lemma}

\begin{proof}[Proof Sketch.]
A univariate real-rooted polynomial with nonnegative coefficients has only nonpositive roots and can be factored as $q(z) = \prod (z+\alpha_i)$ for $\alpha_i \geq 0$. After appropriate scaling, this corresponds to a distribution of a sum of independent Bernoulli random variables.
\end{proof}

\begin{corollary}
If $X \sim \cD$ where $\cD$ is a strongly Rayleigh distribution over subsets of $\cC$, and $F \subset \cC$, then $Z = \sum_{i \in F} X_i$ has a generalized binomial distribution.    
\end{corollary}

We refer the reader e.g. to \cite{Gharan_MaxEntropy} for more details.

\section{Tail bounds for generalized binomial distributions}
\label{sec:tail-bounds}

Our analysis relies crucially on tight tail bounds for the generalized binomial distribution, i.e. summation of independent (but non-identical) Bernoulli variables. Essentially, we prove that the limiting Poisson distribution is the worst case for several tail bounds of interest. We note that similar tail bounds for negatively dependent variables were used in \cite{nguyen-song}; however our bounds are sharper for $\mu \geq 2$ since we work specifically with a generalized binomial distribution. 
We outline the key steps here and prove them in \cref{sec:tail-proofs}. 

In the following, $\on{Ber}(p)$ denotes the Bernoulli random variable, equal to $1$ with probability $p$ or $0$ with probability $1-p$.
$\lfloor \mu \rfloor$ denotes the floor of $\mu$ (largest integer lower-bounding $\mu$). 

\begin{theorem}
\label{thm:main-tail}
Suppose $Y_i\sim \on{Ber}(p_i)$ for $i\in [n]$ are independent, where $p_i\in [0, 1]$ and $\sum_{i=1}^n p_i\ge \alpha \mu$ for some $\mu \ge 0$ and $\alpha \ge 2$. Then $Y=\sum_{i=1}^{n} Y_i$ satisfies
\begin{equation}
\label{eq:tail-eq}
 \mb{P}[Y\le \lfloor \mu \rfloor -1]\le \begin{cases}
0 & \text{if }\lfloor \mu \rfloor=0\\
	e^{-\alpha}&\text{if }\lfloor \mu\rfloor=1\\
	(1+2\alpha)e^{-2\alpha}&\text{if }\lfloor \mu \rfloor\ge 2
\end{cases}\quad\text{and}\quad  \mb{P}[Y\le \lfloor \mu \rfloor -2]\le \begin{cases}
0&\text{if }\lfloor \mu\rfloor\le 1\\
e^{-2\alpha}&\text{if }\lfloor \mu\rfloor\ge 2
\end{cases}
\end{equation}
\end{theorem}
By a simple coupling argument, it suffices to consider the case where $\mu$ is a positive integer and $\sum_{i=1}^n p_i=\alpha \mu$. First, we show that lower tail probabilities in the above form are maximized when $p_i$ takes on at most one fractional value, i.e. $Y \sim a+\on{Bin}(m, p)$ for some $a, m\in\mb{Z}_{\ge 0}$ and $p\in (0, 1)$.
\begin{claim}
\label{claim:tail-smth}
If $Y_i\sim \on{Ber}(p_i)$ for $i\in [n]$ are independent, where $p_i\in [0, 1]$ and $\sum_{i=1}^n p_i=\alpha\mu$, then there exists $a, m\in\mb{Z}_{\ge 0}$ such that $Y=\sum_{i=1}^{n} Y_i$ stochastically dominates $Z\sim a+\on{Bin}(m, (\alpha\mu-a)/m)$.
\end{claim}

Recall that we say $Y$ \emph{stochastically dominates} $Z$ if $\mb{P}(Y\le t)\le \mb{P}(Z\le t)$ for all $t$. Next, we show the bounds hold for $Y\sim \on{Pois}(\alpha \mu)$, heuristically justifying the theorem.
\begin{claim}
\label{claim:tail-pois}
\cref{eq:tail-eq} holds for $Y\sim \on{Pois}(\alpha \mu)$ for any $\mu \in \mb{Z}_{\ge 0}$ and $\alpha \geq 2$.
\end{claim}
The $\mu=0$ and $\mu=1$ bounds in \cref{eq:tail-eq} are clear. For the two bounds in \cref{eq:tail-eq} with $\mu \ge 2$, let
\begin{equation}
\label{eq:flk}
f_\ell(\mu) := \mb{P}\left(\on{Pois}(\alpha \mu)\le \mu-\ell\right)
\end{equation}
and they are equivalent to showing $f_1(\mu)\le f_1(2)$ and $f_2(\mu)\le f_2(2)$ for $\mu \ge 2$, which we do in \cref{sec:tail-proofs}. Actually, one can show that $f_\ell(\mu)$ is decreasing in $\mu$ for $\mu \ge 2$ and $\ell \in \{1, 2\}$. An old result of \cite{teicher,aj} showed the case of $\ell=0$ and $\alpha =1$. This is generalized to $\ell=0$ and $\alpha \ge 1$ in \cite{teicher-2} by studying $f_0(\mu+1)-f_0(\mu)$ as a function of $\alpha$, and bootstrapping the $\alpha=1$ case. However, for $\ell=1, 2$, the case of $\alpha=1$ is false. Nevertheless, it holds for $\alpha\ge 2 $, but for brevity we only show the weaker result that $f_\ell(\mu)$ is maximized at $\mu=2$.

Next, we show the bounds for $Y\sim \on{Bin}(m, p)$ by reducing the binomial case to the Poisson case.

\begin{claim}
\label{claim:tail-bin}
\cref{eq:tail-eq} holds for $Y\sim \on{Bin}(m, p)$ for any $p\in [0, 1]$ and $\mu \in \mb{Z}_{\ge 0}$ with $mp = \alpha \mu$.
\end{claim}

Finally, we put things together to prove \cref{thm:main-tail}.
Using \cref{claim:tail-smth} and \cref{claim:tail-bin}, it remains to show the reduction from the case of $Y\sim a+\on{Bin}(m, (\mu-a)/m)$ to the case where $a=0$ via a coupling argument. This argument as well as the proofs of the claims are deferred to in \cref{sec:tail-proofs}.

\section{Computation and analysis of approximately stable committees}
In this section, we prove the main theorem, \cref{thm:main}, and describe the algorithm.

\subsection{Key Steps in the Algorithm}
First we provide details for the high level sketch of the algorithm in \cref{sec:pf-sketch}, explained line-by-line.
We start with a Lindahl equilibrium $(\bx, \bp)$ of $\ell_1$ norm $k$ (to be chosen later, given a target committee size $K$). Its fractional core-stability  (\cref{thm:lindahl-stable}) motivates us to round to an integer committee that hopefuly preserves the stability property to some extent. We scale $\bx$ by $\alpha \ge 2$ and define $\bx'$ by $x'_i = \min \{ \alpha x_i, 1 \}$. We adjust $\alpha$ so that $\| \bx' \|_1 = \kappa$ is an integer (we return to this issue in more detail later). Then we consider a strongly Rayleigh distribution $\mc{D}$ provided by \cref{thm:sr-us}. 
\begin{definition}
\label{def:v-delta}
Let $n:=|\cV|$. Let $R$ be sampled from the strongly Rayleigh distribution provided by \cref{thm:sr-us}. For $\ell = 1, 2$, let 
\begin{equation}
\label{eq:def-delta}
\cV_\ell : = \left\{ v \in \cV : |R \cap A_v| \le \left\lfloor \sum_{i \in A_v} x_i \right\rfloor - \ell \right\}\quad\text{and}\quad 
    \delta_\ell := \frac{1}{n} |\cV_\ell|,
\end{equation}
i.e.~the voters who prefer $\bx$ to $R$ by at least a certain amount. Clearly, $\mc{V}_2\subseteq \mc{V}_1$ so $\delta_2\le\delta_1$.
\end{definition}

Now, we leverage \cref{thm:main-tail} to argue that in expectation, not too many voters prefer $\bx$ to $R$ significantly. 
By a simple Markov bound, we also obtain that finding the desirable $R$ takes $O(1)$ tries with high probability.
\begin{lemma}
\label{lem:exp-delta}
Using the notation of \cref{eq:def-delta},
$$\mb{E}\left[\delta_1+\left(e^\alpha -1-2\alpha\right)\delta_2\right]\le e^{-\alpha}$$ and for every $\varepsilon>0$ we have
\begin{equation}
\label{eq:exp-cond}
\mb{P}\left[\delta_1+\left(e^\alpha -1-2\alpha\right)\delta_2\ge (1+\varepsilon)e^{-\alpha}\right] \le \frac{1}{1+\varepsilon}
\end{equation}
\end{lemma}
\begin{proof}
Fix a voter $v$, let $F_v:=\{i\in A_v: x_i<1/\alpha\}$, and $C_v:=|A_v\setminus F_v|$, so
\[ \mu_v:=\sum_{i\in F_v} x_i \ge \left(\sum_{i\in A_v}x_i\right) -C_v.\]
Let $Y_v:=\sum_{i\in F_v}X_i$. For each $i\not\in F_v$, $\mb{P}[i \in R] = 1$, so $|R\cap A_v| = C_v + Y_v$. For $i\in F_v$, $\mb{E}X_i = \alpha x_i$, so $\mb{E}Y_v = \alpha \mu_v$.
Now, for $\ell = 1, 2$
\[ \mb{P}[v\in \cV_\ell] = \mb{P}\left[|R\cap A_v| \le \left\lfloor\sum_{i\in A_v}x_i\right\rfloor-\ell \right] \le \mb{P}\left[C+Y_v\le \left\lfloor \mu_v+C_v\right\rfloor-\ell\right] \le \mb{P}[Y_v\le \lfloor \mu_v\rfloor -\ell]. \]

Note that $F_v$ and $\mu_v$ depend only on the instance, $\bx$, $\alpha$, and $v \in \cV$. In particular, there is no dependence on $R$.
By \cref{thm:sr-us}, $Y_v$ is a sum of independent Bernoulli's (with possibly different means $p_i\ne \mb{E}X_i$, but the same total expectation $\sum_{i\in F_v} p_i =\mb{E}Y_v = \alpha \mu_v$). We can apply \cref{thm:main-tail} to bound $\mathbb{P}[v\in \cV_1]$ and $\mathbb{P}[v\in \cV_2]$. 
We distinguish two cases.

If $\lfloor \mu_v \rfloor \leq 1$, then
$$ \mb{P}[v \in \cV_1] \leq  e^{-\alpha}, \ \ \
\mb{P}[v \in \cV_2] \leq 0. $$

If $\lfloor \mu_v \rfloor \geq 2$, then
$$ \mb{P}[v \in \cV_1] \leq  (1+2\alpha) e^{-2\alpha}, \ \ \
\mb{P}[v \in \cV_2] \leq e^{-2\alpha}. $$
Either way, we have
$$ \mb{P}[v \in \cV_1] + (e^\alpha - 1 - 2\alpha) \mb{P}[v \in \cV_2] \leq e^{-\alpha}.$$
The lemma follows by adding up over all voters,
and \cref{eq:exp-cond} holds by Markov's inequality.
\end{proof}

Fix any $R$ attaining the inequality in \cref{lem:exp-delta}. Initialize $\cV':=\cV_1\setminus \cV_2$. We can reparametrize
\begin{equation}
\label{eq:delta}
t:=\delta_2 \qquad \text{and}\qquad \delta_1\le f(t):= e^{-\alpha}-(e^\alpha -1-2\alpha)t.
\end{equation}
Therefore, as $\mc{V}_2\subseteq\mc{V}_1$, we have $0\le t = \delta_2\le\delta_1\le f(t)$.
Denote 
\begin{equation}
\label{eq:domain-t}
 t_0 :=\frac{e^{-\alpha}}{e^\alpha-2\alpha}
\end{equation}
and note that $f(t_0) = t_0$. Also, $f$ is decreasing in $t$, and $f(t) < t$ for $t > t_0$.
Therefore, $t > t_0$ cannot occur since that would mean $\delta_1 < \delta_2$.

At Line 6 of the main algorithm, for some choice of $\beta, \gamma$ possibly depending on $t$,  we iteratively find additional candidates approved by at least $\beta n/k$ voters in $\cV'$, and delete those voters from $\cV'$. We denote the additional candidates $R'$ and the voters removed for candidate $i \in R'$ by $V_i$. We stop when either (Case 1) we have found $\lceil \gamma k \rceil$ such candidates, or (Case 2) no more such candidates can be found.
\begin{lemma}[Case 1]
\label{lem:case-1}
Suppose there is a set $R'\subset \cC\setminus R$ of $\lceil \gamma k\rceil$ candidates, and disjoint sets of voters $\{V_i\subset \cV':i\in R'\}$ such that $i\in A_v$ for all $v\in V_i$. Then for every $T\subset \cC$
\[ \left|\left\{v\in (\cV\setminus \cV_1)\cup \bigcup_{i\in R'}V_i: |A_v\cap T|> |A_v\cap (R\cup R')|\right\}\right|\le |T| \frac{n}{k}.\]
\end{lemma}
\begin{proof}
If $v\in \cV\setminus \cV_1$, then by definition of $\cV_1$ \[ |A_v\cap T| >|A_v\cap (R\cup R')|\ge |A_v\cap R|\ge \left\lfloor \sum_{j\in A_v}x_j\right\rfloor\]
If $v\in V_i$, then $i\in A_v\cap (R'\setminus R)$, and $v\in \cV_1\setminus \cV_2$ so
\[  \left\lfloor \sum_{j\in A_v}x_j\right\rfloor-1=|A_v\cap R|< |A_v\cap (R\cup R')|\le |A_v\cap T|\]
Either way, $v$ prefers $T$ to $\mbf{x}$, so there are less than $|T| \frac{n}{k}$ such voters in total by \cref{thm:lindahl-stable}. 
\end{proof}

\begin{lemma}[Case 2]
\label{lem:case-2}
If $R'\subset \cC\setminus R$ is a maximal set of candidates for which we can find disjoint sets of voters $\{V_i\subset \cV':i\in R'\}$ where $i\in A_v$ for all $v\in V_i$,
and $|R'|<\lceil \gamma k\rceil$, then for every $T\subset \cC$
\[ \left|\{v\in \cV\setminus \cV_2: |A_v\cap T|> |A_v\cap (R\cup R')|\}\right|\le (1+\beta)|T| \frac{n}{k}.\]
\end{lemma}
\begin{proof}
The number of such $v\in \cV\setminus \cV_1$ is at most $|T| \frac{n}{k}$ similar to \cref{lem:case-1}. The number of such $v\in \cV_1\setminus \cV_2$ is at most the number of pairs of the form
\[ \{(v, i)\in (\cV_1\setminus \cV_2) \times (T\setminus (R\cup R')):i\in A_v\}\]
because we cannot have $|A_v \cap T| > |A_v \cap (R \cup R')|$ if $v$ does not approve any candidate in $T \setminus (R \cup R')$.
Each candidate $i \in T \setminus (R \cup R')$ contributes at most $\beta \frac{n}{k}$ such pairs by the maximality of $R'$, hence there are at most $\beta |T| \frac{n}{k}$ such pairs.
\end{proof}

On Line 7 of the main algorithm, we recurse on the complement of the voters in \cref{lem:case-1,lem:case-2} respectively. Let $K$ be the total committee size. We delete the chosen candidates $R\cup R'$ from the approval sets and find recursively a $\lambda$-approximately stable committee of size $K-(\alpha+\gamma)k$. 

Note that in the randomized rounding in \cref{thm:sr-us} and \cref{lem:case-1}, we technically need integer parameters $\kappa = \lceil \alpha k\rceil$ and $\lceil \gamma k\rceil$.
For now, let us assume that this incurs $2r = O(\log B)$ additional candidates where $r$ is the number of recursive calls of the algorithm. We point out that the same rounding issue seems to arise in \cite{munagala2021approximate,nguyen-song}, and so the claimed factors hold only asymptotically as $k \to \infty$.  We will return to this issue in \cref{sec:full-algo}, where we prove a bound valid for all $k$.

We analyze the final approximation ratio $\lambda$ of this algorithm as a numerical optimization program. We first justify the correctness of this program and then check the feasibility of \cref{eq:program}.

\begin{lemma}
\label{lem:program}
For any $\alpha \ge 2, \eta \in (0, 1)$, $\lambda > 0$, let $f(t) = e^{-\alpha}-(e^\alpha-1-2\alpha) t \ge 0$. Suppose that there exist functions $\beta(t), \gamma(t) \ge 0$ such that for $t \geq 0$,
\begin{equation}
\label{eq:program}
\begin{aligned}
    \lambda_1&:= \frac{1}{\eta}+\frac{\lambda (f(t)-\beta(t) \gamma(t))}{1-(\alpha+\gamma(t))\eta}\le \lambda \\
    \lambda_2&:= \frac{1+\beta(t)}{\eta}+\frac{\lambda t}{1-(\alpha+\gamma(t))\eta}\le \lambda
\end{aligned}
\end{equation}
Then, running the main algorithm with parameter $k=\lceil \eta K\rceil$ will output a committee $S$ of size at most $K+O(\log K)$. Moreover, for every $T \subset \cC$ we have
\[ |\{v\in\cV:|A_v\cap S|<|A_v\cap T|\}|\le \frac{\lambda |T|n}{K}\]
\end{lemma}

In the following, we drop the $t$ argument from $\beta(t), \gamma(t)$ to make the notation less cluttered.

\begin{proof}
Induct on the number of candidates. In Case 1, let $\cV'' = \cV_1\setminus \bigcup_{i\in R'}V_i$. By \cref{lem:case-1}, the deviating coalition of voters in $\cV\setminus \cV''$ who prefer $T$ to $R\cup R'$ has size at most $n|T|/k$. The number of remaining voters for the recursive call is
\[ |\cV''| \leq |\cV_1| - |R'|\cdot \frac{\beta n}{k} \le  (\delta_1-\beta\gamma)n, \]
and the size of the remaining committee to be found is $K - (\alpha+\gamma) k$, so by the inductive hypothesis the deviating coalition that prefers $T$ to its output has size at most $\lambda |T|\cdot |\cV''|/(K-(\alpha+\gamma)k)$. The total deviating coalition is bounded as
\[ |\{v\in\cV:|A_v\cap S|<|A_v\cap T|\}|\le  \frac{n|T|}{k}+\frac{\lambda |T|(\delta_1-\beta\gamma)n}{K-(\alpha+\gamma)k}\le  \underbrace{\left(\frac{1}{\eta}+\frac{\lambda(f(t)-\beta\gamma)}{1-(\alpha+\gamma)\eta}\right)}_{\lambda_1}\frac{|T|n}{K}\]

Similarly in Case 2, by \cref{lem:case-2} and the inductive hypothesis, the deviating coalition has size
\[ |\{v\in\cV:|A_v\cap S|<|A_v\cap T|\}|\le \frac{(1+\beta)n|T|}{k}+\frac{\lambda |T|\delta_2n}{K-(\alpha+\gamma)k} = \underbrace{\left(\frac{1+\beta}{\eta}+\frac{\lambda t}{1-(\alpha+\gamma)\eta}\right)}_{\lambda_2}\frac{|T|n}{K}\]
Together, \cref{eq:program} guarantees the overall approximation ratio $\lambda$, completing the induction.
\end{proof}

We then solve the program \cref{eq:program} to minimize approximation ratio $\lambda$.
\begin{lemma}
\label{lem:parameters}
\cref{eq:program} is feasible for $\lambda = 3.606655$ by choosing
$(\alpha, \eta, \gamma)= (2.154564, 0.35870, 0.30328)$ and $\beta(t) = (\lambda\eta -1)(1-t/t_0)$ where $t_0$ is defined in \cref{eq:domain-t}.
\end{lemma}

\begin{proof} 
For any $\alpha, \lambda$, we choose 
$\gamma, \eta$ such that
\begin{equation}
\label{eq:beta-gamma}
\lambda = \frac{1}{\eta}+\frac{\lambda t_0}{1-(\alpha+\gamma)\eta}\quad\text{and}\quad \gamma = \frac{t_0(e^\alpha-1-2\alpha)}{\lambda\eta-1}
\end{equation}
We then minimize $\lambda$ numerically over $\alpha\ge 2$ to obtain $\alpha =2.154564$ and $ \lambda = 3.606655$. Then, \cref{eq:beta-choice} solves numerically to $\gamma = 0.30328$ and $\eta = 0.35870$. We show this is feasible for \cref{eq:program} where
\begin{equation}
\label{eq:beta-choice}
\beta(t) = \left(\lambda\eta-1\right)\left(1-\frac{t}{t_0}\right) = \frac{\lambda \eta(t_0-t)}{1-(\alpha+\gamma)\eta}\in [0, 1]
\end{equation}
using \cref{eq:beta-gamma} and that $t\in [0, t_0]$. Then we compute
\[ \gamma \cdot \beta(t)= \left(\lambda\eta-1\right)\left(1-\frac{t}{t_0}\right) \left(\frac{t_0(e^\alpha-1-2\alpha)}{\lambda\eta-1}\right) = (t_0-t) (e^\alpha-1-2\alpha)=f(t)-f(t_0)=f(t)-t_0\]
and therefore
\begin{equation}
    \begin{aligned}
    \lambda_1 &= \frac{1}{\eta} +\frac{\lambda (f(t)-\beta(t)\gamma(t))}{1-(\alpha+\gamma(t))\eta} =  \frac{1}{\eta} +\frac{\lambda t_0}{1-(\alpha+\gamma)\eta} = \lambda\\
    \lambda_2 &= \frac{1+\beta(t)}{\eta}+\frac{\lambda t}{1-(\alpha+\gamma(t))\eta} = \frac{1}{\eta}+\frac{\lambda t_0}{1-(\alpha+\gamma)\eta}
    =\lambda
    \end{aligned}
\end{equation}
using the identities for $\beta(t)$ and $\gamma \cdot \beta(t)$ from above.
\end{proof}

We remark that $\lambda = 3.606655$ is (roughly) the minimal feasible value of \cref{eq:program} even when $\gamma$ is allowed to be a function of $t$. We chose $\gamma $ to be constant for clarity, and maximized for the given $\alpha$ and $\lambda$ to minimize the number of recursive calls, thereby minimizing the rounding error. This choice is optimal by considering $t=t_0$ and noting $\beta(t) \ge 0$.

\begin{algorithm}
\DontPrintSemicolon
\caption{Detailed Main Algorithm}
\label{alg:alg-1}
\KwIn{An ABC instance $(\cV, \cC, \{A_v : v \in \cV\})$ and $K \le |\cC|$}
\KwOut{A committee $S$ of size $K + O(\log K)$
}
\If{$K \le 28$}{
   \Return Committee given by \cref{cor:base-case}
}
$\alpha \gets 2.154564$\; 
$\eta \gets 0.358696$\;
$\varepsilon \gets 10^{-10}$\;
$k \gets \lceil \eta K \rceil$\;
Compute a Lindahl equilibrium $\mbf{x}$ with $\|\mbf{x}\|_1 = k$ as in \cref{thm:lindahl-compute}\;
Construct a strongly Rayleigh distribution $\cD$ on $2^{\cC}$ as in \cref{thm:sr-us}, such that $\mathbb{P}_{R \sim \cD}[i \in R] \geq \min \{\alpha x_i, 1\}$ and $|R| = \lceil \alpha k\rceil $\;

\For{$\ell = 0,1$}{
    $\cV_\ell \gets \left\{ v \in \cV : |R \cap A_v| \le \left\lfloor \sum_{i \in A_v} x_i \right\rfloor - \ell \right\}$\;
    $\delta_\ell \gets |\cV_\ell| / n$\;
}

Sample $R \sim \cD$ until
$\delta_1 + (e^{\alpha} - 1 - 2\alpha)\delta_2 \le (1+\varepsilon)e^{-\alpha}$\;
$R' \gets \emptyset$\;
$\cV' \gets \cV_1 \setminus \cV_2$\;
Choose $\beta, \gamma$ based on \cref{lem:parameters}\;
\While{$|R'| \le \lceil \gamma k \rceil$ \textbf{and} there exists $i \in \cC \setminus (R \cup R')$ with $|\{v \in \cV' : i \in A_v\}| \ge \beta k / n$}{
    $R' \gets R' \cup \{i\}$\;
    $\cV' \gets \cV' \setminus \{v \in \cV' : i \in A_v\}$\;
}
\If{$|R'|=\lceil \gamma k\rceil$}{$\cV''\gets \cV_2\cup \cV'$}
\Else{$\cV''\gets \cV_2$}
$\cC''\gets  \cC\setminus (R\cup R')$\;
$R''\gets \text{Algorithm 1 }(\cV'', \cC'', \{A_v\cap \cC'':v\in \cV''\}, K-(\alpha+\gamma)k)$\;
\Return $R \cup R' \cup R''$\;
\end{algorithm}

\begin{algorithm}
\DontPrintSemicolon
\caption{$3.651$-Approximately Stable Committee Selection for ABC Elections}
\label{alg:alg-1}

\KwIn{An ABC instance $(\cV, \cC, \{A_v : v \in \cV\})$ and positive integer $K\le |\cC|$}
\KwOut{A $3.651$-approximately core stable committee of size $K$}
$\rho\gets 0.00703$\;
$K' \gets K/(1+2\rho)$\;
$S\gets \text{Algorithm 1 }(\cV, \cC, \{A_v:v\in \cV\}, K')$\;
$S'\gets $ any $K-|S|$ candidates $S'$ from $\cC\setminus S$\;
\Return $S\cup S'$
\end{algorithm}

\subsection{The Full Algorithm}
\label{sec:full-algo}
In this section, we give the full main algorithm and deal with the rounding issue rigorously; this increases the approximation ratio from $3.606$ to $3.651$. Let a positive integer $K \le |\cC|$ be the target committee size. As our algorithm is recursive, we first consider the base case when $K$ is small.
\begin{theorem}[\cite{peters2025core}]
\label{thm:pav}
If $K\le 8$, there is an exactly stable committee of size $K$.
\end{theorem}
\begin{corollary}
\label{cor:base-case}
If $K\le 28$, there is always a $3.5$-approximately stable committee of size $K$ computable in polynomial time.
\end{corollary}
\begin{proof}
There is an exactly stable committee $S$ of size $\min(K, 8)$ (based on the PAV rule) by  \cref{thm:pav}. We can compute this in $O(|\cC|^{16})$ time by enumerating through all $O(|\cC|^8)$ possible $S$ and checking against all $O(|\cC|^8)$ possible smaller blocking committees.

We add $K-\min(K,8)$ candidates arbitrarily to get a committee $S'$ of size $K$. Then the deviating coalition of any $T$ is at most
\[ \frac{|T|n}{|S|} = \frac{|T|n}{\min(K, 8)}\le \frac{3.5|T|n}{K}.\]
\end{proof}

The number of recursive calls of Algorithm 1 is
$r=\log_{1-(\alpha+\gamma)\eta} \left(28/{K'}\right) $ with the parameters in \cref{lem:parameters}. The following lemma formalizes the analysis of the rounding error at the expense of the approximation ratio.
\begin{lemma}
\label{lem:recur}
The output $S$ of invoking Algorithm 1 on Line 3 of Algorithm 2 satisfies $|S|\le K'+2r \le K$.
\end{lemma}
\begin{proof}
$|S|\le K'+2r$ by \cref{lem:program} where $r=\log_{1-(\alpha+\gamma)\eta} \left({28}/{K'}\right)$. As $1-(\alpha+\gamma)\eta\le 0.15405$ and \[ g(t)=\frac{1}{t}\log_{0.15405}\left(\frac{28}{t}\right)\]
is maximized over $t\in \mb{N}$ at $g(76)\le 0.00703 =\rho$, we have $|S|\le K'
 (1+2\rho)=K$. 
\end{proof}
\begin{theorem}
\label{thm:correctness}
Algorithm 2 has expected polynomial running time and outputs a $3.651$-approximately stable committee for any given size $K \in \mb{N}$.
\end{theorem}

\begin{proof}
On line 12 of Algorithm 1, it takes $O(1)$ expected samples to find a good set $R$ by \cref{lem:exp-delta}. As $\varepsilon$ is a sufficiently small constant, the error from $1+\epsilon$ in the inequality on line 12 can be absorbed into the error incurred by rounding up to four or five decimal places everywhere, so we disregard $\varepsilon$. The rest of the algorithm runs in polynomial time as the base case does by \cref{cor:base-case} and there are $O(\log K)$ recursive steps.

For correctness, note that by \cref{lem:program,lem:parameters}, upon running the main algorithm, we find a committee $S$ of size $K'+2r$ such that the deviating coalition for any alternative committee $T\subset \mc{C}$ is at most
\[ \frac{3.6067|T|n}{K'}\le  \frac{3.6067(1+2\rho)|T|n}{K}\le\frac{3.651|T|n}{K}\]
By \cref{lem:recur}, $|S|\le K'+2r \leq K$ so we always output a committee of size exactly $K$.
\end{proof}


\bibliographystyle{amsplain0.bst}
\bibliography{main.bib}

\appendix

\section{Proofs of Tail Bound for Generalized Binomial Distributions}
\label{sec:tail-proofs}
Here we prove the claims from \cref{sec:tail-bounds}.
\begin{proof}[Proof of \cref{claim:tail-smth}]
We fix any $t\in\mb{Z}_{\ge 0}$ and show $\mb{P}(Y\le t)\le \mb{P}(Z\le t)$.
Let $(p_i)_{i=1}^n$ be the maximizer of $\mb{P}(Y\le t)$ subject to $\sum_{i=1}^n p_i=\mu$ and $p_i \in [0,1]$, which exists by compactness.
Let $A=\{i\in[n]:p_i=1\}$ with $a=|A|$ and $B=\{i\in [n]:p_i=0\}$. Without loss of generality, we assume that $[n]\setminus (A\cup B) = [m]$, and let $Y' = \sum_{i=1}^{m}Y_i$, so $Y=Y'+a$ with probability one.

If $m=0$ or $m=1$, $Y$ can already be written in the form of $a+\on{Bin}(m, p)$, and we are done.
Otherwise, $(p_1, \dots, p_{m})$ maximizes $\mb{P}(Y'\le t-a)$. Pick any two indices $i, j\in [m]$.  
\begin{align*}
\mb{P}(Y'=s) & =\sum_{S\in \binom{[m]}{s}}\prod_{\ell\in S}p_\ell\prod_{\ell\not\in S\cup\{i, j\}}(1-p_\ell)
\\ & =p_{i}p_{j}\sum_{S\in\binom{[m]\setminus{\{i, j\}}}{s-2}}\prod_{\ell\in S}p_\ell\prod_{\ell\not\in S\cup\{i, j\}}(1-p_\ell)
\\ &\quad + [p_{i}(1-p_j)+p_j(1-p_i)]\sum_{S\in\binom{[m]\setminus\{i, j\}}{k-1}}\prod_{\ell\in S}p_\ell\prod_{\ell\not\in S\cup\{i, j\}}(1-p_\ell)
\\ & \quad + (1-p_i)(1-p_j)\sum_{S\in \binom{[m]\setminus\{i, j\}}{s}}\prod_{\ell\in S}p_\ell\prod_{\ell\not\in S\cup\{i, j\}}(1-p_\ell)
\end{align*}
Viewing the other $m-2$ many $p_\ell$'s as fixed constants, this is a symmetric, quadratic, and multilinear function in $(p_i, p_j)$. Therefore, $(p_i, p_j)$ maximizes A symmetric, quadratic, and multilinear function $\mb{P}(Y'\le t-a)$ subject to the constraint that $p_i+p_j$ is fixed. Maximizing such of function is equivalent then to maximizing or minimizing $p_ip_j$. The only possible optimum where $p_i, p_j \not\in \{0, 1\}$ is $p_i=p_j$. This holds for every pair of indices $i, j$ where $p_i, p_j\not\in \{0, 1\}$, so all fractional $p_i$'s must equal some $p\in (0, 1)$. This is precisely the distribution of $Z = a + Bin(m,p)$, where the choice of $p=(\mu-a)/m$ corresponds to the fixed sum constraint that $\sum_{i=1}^n p_i = \mu$.
\end{proof}

\begin{proof}[Proof of \cref{claim:tail-pois}]
The bounds for $\mu=0$ and $\mu=1$ are clear. It suffices to show $f_1(\mu)\le f_1(2)$ and $f_2(\mu)\le f_2(2)$ for $\mu\ge 2$ where $f_\ell(\mu)$ is defined in \cref{eq:flk}.
Let $s=\mu-\ell$. By the correspondence of lower tails of the Poisson distributions and upper tails of the Gamma distribution, we have that
\[f_\ell(\mu) = \frac{1}{s!}\int_{\alpha \mu}^\infty x^se^{-x}dx.\]
Let $g_s(x)=-\log(x^se^{-x}) = x-s\log x$. Then, $g_s'(x)=1-s/x$ and $g_s''(x)=s/x^2\ge 0$. Hence 
\[ g_s(x)\ge g_s(\alpha \mu)+g_s'(\alpha \mu)(x-\alpha \mu)\]
for $x\ge \alpha \mu$. Therefore, using $s!\ge s^se^{1-s}$, $s\le \mu$, and $\alpha \ge 2$, we have
\[ f_\ell(\mu)\le \frac{1}{s!}\int_{\alpha \mu}^\infty (\alpha \mu)^se^{-\alpha \mu}e^{(-1+s/\alpha \mu)(x-\alpha \mu)}dx \le \frac{(\alpha \mu)^se^{-\alpha \mu}}{s!(1-s/\alpha \mu)}\le \frac{e^{s-1}(\alpha \mu)^{s+1}e^{-\alpha \mu}}{s^s(\alpha \mu -s)}. \]
When $\ell=1$ and $s=\mu-1$, using $(1+1/(\mu-1))^{\mu-1}\le e$, we have
\[ f_1(\mu)\le \frac{e^{\mu-2}(\alpha \mu)^\mu e^{-\alpha \mu}}{(\mu-1)^{\mu-1}(\alpha \mu - \mu+1)}\le \frac{e^{\mu-1}\mu\alpha^\mu e^{-\alpha \mu}}{\alpha \mu-\mu+1}\le \frac{\mu\left(\alpha e^{1-\alpha}\right)^\mu}{e(\mu+1)}\]
We check for $\mu\ge 4$ that $(\alpha e^{1-\alpha})^\mu e^{2\alpha}/(1+2\alpha)$, which is the right-hand side divided by $f_1(2)$ up to a constant factor, is decreasing for $\alpha \ge 2$, and the right-hand side is at most $f_1(2)$ for $\mu=4$, so $f_1(\mu)\le f_1(2)$ for $\mu\ge 4$. We check $f_1(3)\le f_1(2)$ below.

\begin{center}
\includegraphics[height = 4cm]{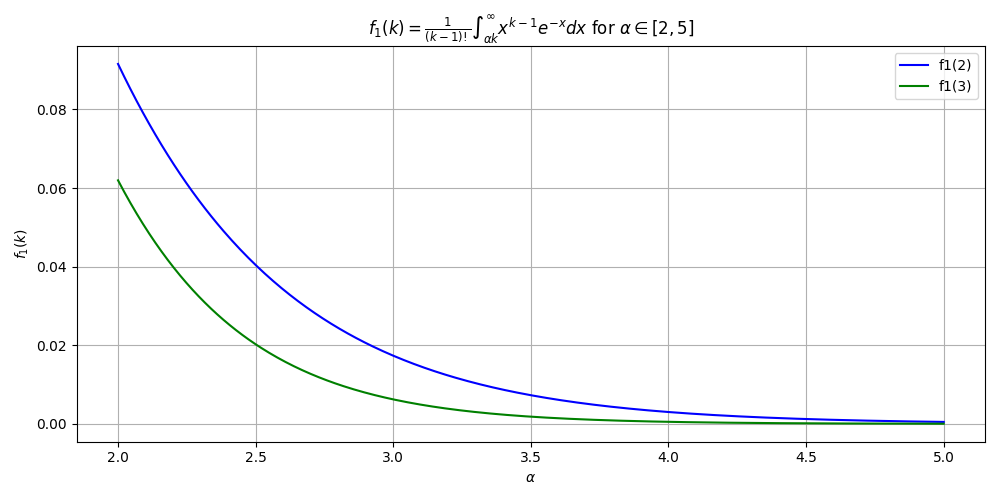}
\, 
\includegraphics[height = 4cm]{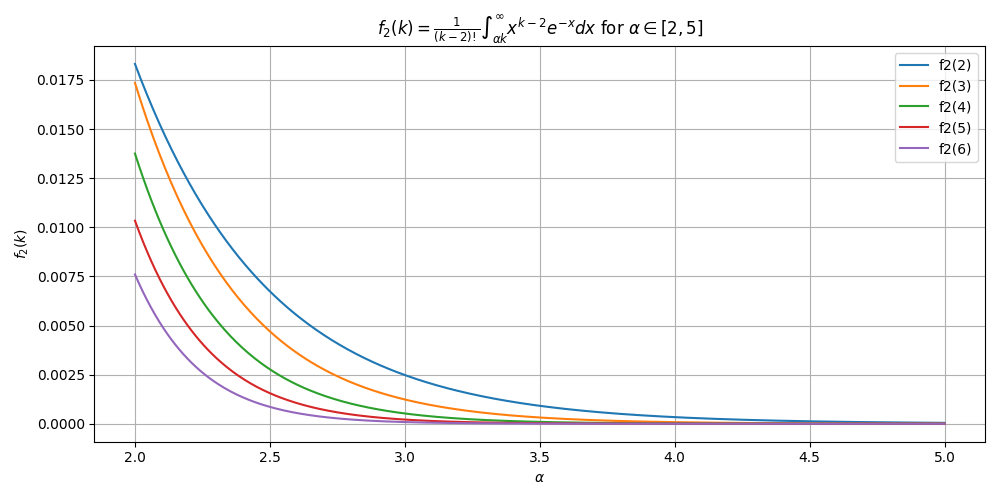}
\end{center}

When $\ell = 2$ and $s=\mu-2$, using $(1+2/(\mu-2))^{\mu-2}\le e^2$, we have
\[ f_2(\mu)\le \frac{e^{\mu-3}(\alpha \mu)^{\mu-1}e^{-\alpha \mu}}{(\mu-2)^{\mu-2}(\alpha \mu - \mu +2)}\le \frac{e^{\mu-1}\mu\alpha ^{\mu-1}e^{-\alpha \mu}}{\alpha \mu-\mu+2}\le \frac{\mu\left(\alpha e^{1-\alpha}\right)^\mu}{2e(\mu+2)}\]
We check for $\mu\ge 4$ that $(\alpha e^{1-\alpha})^\mu e^{2\alpha}$, which is the right-hand side divided by $f_2(2)$ up to a constant factor, is decreasing for $\alpha \ge 2$, and the right-hand side is at most $f_2(2)$ for $\mu=7$, so $f_1(\mu)\le f_1(2)$ for $\mu\ge 7$. We check $f_1(\mu)\le f_1(2)$ for $\mu=3, 4, 5, 6$ above.
\end{proof}

\begin{proof}[Proof of \cref{claim:tail-bin}]
For any $\ell \le \mu $ (in particular $\ell=\mu -1, \mu -2$), we bound
\[\mb{P}(Y=\ell) = \binom{\alpha \mu /p}{\ell} p^\ell (1-p)^{\alpha \mu /p-\ell} \le \frac{( \alpha \mu )^\ell}{\ell!}f_{\alpha, \mu }(p)\quad\text{where}\quad f_{\alpha, \mu }(p)=(1-p)^{\alpha\mu /p-\mu }. \]
We claim that $f_{\alpha, \mu} '(p)\le 0$ for $p\in [0, 1]$. We compute
\[f_{\alpha, \mu }'(p)  = -\frac{\alpha \mu }{p^2}\log (1-p)-\left(\frac{\alpha \mu }{p}-\mu \right)\frac{1}{1-p} \]
so it suffices to show $g(\alpha, p)=\alpha p^{-2}\log (1-p)+(\alpha/p-1)/(1-p)\ge 0$. Finally, we verify that 
\[ \frac{\partial g}{ \partial\alpha} (\alpha, p)= \frac{\ln\left(1-p\right)}{p^{2}}+\frac{1}{p\left(1-p\right)}\ge 0\quad \text{and}\quad g(2, p)=\frac{2}{p}\ln\left(1-p\right)+\frac{2-p}{1-p}\ge 0 \]
for $p\ge 0$, so $g(\alpha, p)\ge g(2, p)\ge 0$ for all $\alpha\ge 2$. Hence, $f_{\alpha, \mu }(p)\le 0$ for all $\alpha \ge 2$ and $p\in[0, 1]$, so
\[ f_{\alpha, \mu }(p)\le \lim_{p\to 0^+} f_{\alpha, \mu }(p) =e^{-\alpha \mu }\implies \mb{P}(Y = \ell) \le e^{-\alpha \mu } \frac{(\alpha \mu )^\ell}{\ell!}\]
This establishes the Poisson approximation as an upper bound of the probability mass at $\ell\le \mu $. For this claim, $\ell=\mu -1$ and $\ell=\mu -2$, so applying \cref{claim:tail-pois} suffices.
\end{proof}
\begin{proof}[Proof of \cref{thm:main-tail}]
By a simple coupling argument, it suffices to show the case where $\mu $ is a positive integer and $\sum_{i=1}^n p_i=\alpha \mu $. Apply \cref{claim:tail-smth} to reduce to the case $Y\sim a+\on{Bin}((\alpha \mu -a)/p, p)$ for some $a\in\mb{Z}_{\ge 0}$ and $p\in [0, 1]$.
The $a=0$ case holds by \cref{claim:tail-bin}. It suffices to show the case for any $a\ge 1$, let $\mu '=\mu -a$. Then, by the $a=0$ case applied to $\mu '$
\begin{align*}
\mb{P}\left(Y\le \mu -1\right) & = \mb{P}\left(a+\on{Bin}\left(\frac{\alpha \mu -a}{p}, p\right)\le \mu -1\right)
\\ & =\mb{P}\left(\on{Bin}\left(\frac{\alpha \mu '+a(\alpha-1)}{p}, p\right)\le \mu '-1\right)
\\ & \le \mb{P}\left(\on{Bin}\left(\frac{\alpha \mu '}{p}, p\right)\le \mu '-1\right)
\end{align*}
If $\mu '\ge 2$, this is at most $(2\alpha+1)e^{-2\alpha}$ by \cref{claim:tail-bin}, as desired. If $\mu '=0$, the probabilities are all zero, so the bound also holds trivially. If $\mu '=1$, we know by the inequality $1-x\le e^{-x}$ that
\begin{align*}
\mb{P}\left(Y\le \mu -1\right) & =\mb{P}\left(\on{Bin}\left(\frac{\alpha +a(\alpha-1)}{p}, p\right)= 0\right) \\
& = (1-p)^{(\alpha+a(\alpha-1))/p}
\\ & \le e^{-\alpha-a(\alpha-1)}
\\ & = e^{1-2\alpha-(\alpha-1)(a-1)}
\\ &\le (1+2\alpha)e^{-2\alpha}
\end{align*}
since $e\le 1+2\alpha$ and $a, \alpha \ge 1$. This proves the first bound in \cref{thm:main-tail}. For the second bound\begin{align*}
\mb{P}\left(Y\le \mu -2\right) & = \mb{P}\left(a+\on{Bin}\left(\frac{\alpha \mu -a}{p}, p\right)\le \mu -2\right)
\\ & =\mb{P}\left(\on{Bin}\left(\frac{\alpha \mu '+a(\alpha-1)}{p}, p\right)\le \mu '-2\right)
\\ & \le \mb{P}\left(\on{Bin}\left(\frac{\alpha \mu '}{p}, p\right)\le \mu '-2\right).
\end{align*}
If $\mu '\ge 2$, this is at most $e^{-2\alpha}$ by \cref{claim:tail-bin}, as desired. If $\mu '\le 1$, the probabilities are all zero, so the bound holds trivially. This proves \cref{thm:main-tail} entirely.
\end{proof}

 


\end{document}